\documentclass[letterpaper, 10 pt,conference]{ieeeconf}  % Comment this line out if you need a4paper

\IEEEoverridecommandlockouts                              % This command is only needed if 
\usepackage{graphicx} % for pdf, bitmapped graphics files
\usepackage{epsfig} % for postscript graphics files
\usepackage{cite} % allows to put citations like [1]-[5]
\usepackage{amssymb,amsmath,graphicx}
\usepackage{mathrsfs}
\usepackage{mathtools}
\usepackage{color}
\usepackage{bm}
\usepackage{nicefrac}
\let\IEEEproof\proof                         % make \IEEEproof do same as \proof
\let\proof\relax

\usepackage{amsthm}
\usepackage{algorithm, algpseudocode}  
\graphicspath{{.}}

\theoremstyle{plain}

\usepackage{etoolbox}
\usepackage{float}
\usepackage{tabularx}

\def\th@definition{
	\thm@headfont{\itshape} % Heading font is italic
	\thm@notefont{} % Note is same as heading
}
\makeatother

\theoremstyle{definition}

\newtheorem{dfn}{Definition}
\newtheorem{lemma}{Lemma}

\newtheorem{problem}{Problem}
\newtheorem{theorem}{Theorem}
\newtheorem{remark}{Remark}

\newtheorem{proposition}{Proposition}

\newcommand{\barsig}{\overline{\Sigma}}
\newcommand{\tell}{\tilde{\ell}}

\newcommand{\adj}{\textnormal{Adj}}

\makeatletter
\patchcmd{\@makecaption}
{\scshape}
{}
{}
{}
\makeatother

\newcommand{\norm}[1]{\left\lVert#1\right\rVert}

\let\oldIEEEkeywords\IEEEkeywords
\def\IEEEkeywords{\oldIEEEkeywords\normalfont\bfseries\ignorespaces}
\newcommand{\argmin}{\mathrm{arg}\min}

\begin{document}
	
\title{Differentially Private Controller Synthesis With Metric Temporal Logic Specifications}           
                                 
	\author{Zhe~Xu, Kasra Yazdani, Matthew T. Hale and Ufuk Topcu
	\thanks{Zhe~Xu is with the Oden Institute
		for Computational Engineering and Sciences, University of Texas,
		Austin, Austin, TX 78712, Kasra Yazdani and Matthew T. Hale are with the Department of Mechanical
		and Aerospace Engineering, University of Florida, Gainesville, Florida 32611, Ufuk Topcu is with the Department
		of Aerospace Engineering and Engineering Mechanics, and the Oden Institute
		for Computational Engineering and Sciences, University of Texas,
		Austin, Austin, TX 78712, e-mail: zhexu@utexas.edu, kasra.yazdani@ufl.edu, matthewhale@ufl.edu, utopcu@utexas.edu.}     
}

\maketitle                 

\begin{abstract} 
    Privacy is an important concern in various multi-agent systems in which data collected from the agents are sensitive. We propose a \textit{differentially private} controller synthesis approach for multi-agent systems subject to high-level specifications expressed in \textit{metric temporal logic} (MTL). We consider a setting where each agent sends data to a \textit{cloud} (computing station) through a set of \textit{local hubs} and the cloud is responsible  for computing the control inputs of the agents. Specifically, each agent
    adds \textit{privacy noise} (e.g., Gaussian noise) point-wise in time to its own outputs before sharing them with a
   \textit{local hub}. Each local hub runs a Kalman filter to estimate the state of the corresponding agent and periodically sends such state estimates to the cloud. The cloud computes the optimal inputs for each agent subject to an MTL specification. While guaranteeing differential privacy of each agent, the controller is also synthesized to ensure a probabilistic guarantee for satisfying the MTL specification. We provide an implementation of the proposed method on a simulation case study with two Baxter-On-Wheels robots as the agents.
\end{abstract}       
 
\section{Introduction}           
\label{sec:intro}
	Along with the rapid development of multi-agent systems (MAS) and cloud computing technologies, protecting the privacy of collected data has been a major concern \cite{zhe2019privacy}. While the detailed available data from the agents in an MAS helps in decision-making and control, the resolution of the shared data triggers the possibility of breaching the privacy of the agents. For example, while smart transportation systems rely on precise measurements of locations of vehicles, the shared information can be sensitive as it may reveal traces of movements of vehicles.
	% in deployment of such systems Privacy is an important concern in many distributed systems in which data collected from entities is sensitive.  For example, while precise measurements for power consumption of users of the power gird provides certain services to the smart grid, the shared information can reveal various activities or even daily schedule of a consumer \cite{molina-Markham2010}. 
	%	Motivated by the threats on privacy, the existing literature provides methods such as $k$-anonymity \cite{sweeney2002}, homeomorphic encryption, information theoretic privacy \cite{sankar2013}, and differential privacy to preserve privacy \cite{Dwork2014algorithmic}.

	Differential privacy constitutes a strong standard for protecting the privacy of agents while allowing for general statistical analyses on aggregate data \cite{Dwork2014algorithmic}. Differential privacy was originally developed for static data and it provides several important properties such as resilience to post-processing \cite{Dwork2014algorithmic}. Besides, privacy guarantees of differential privacy hold against adversaries with auxiliary information that could potentially be linked with sensitive data of agents. 
	
    More recently, the guarantees of differential privacy have been extended to dynamical systems in which trajectory-valued data are protected \cite{LeNy2014}. Differential privacy for trajectory-valued data is achieved by adding \textit{privacy noise} (e.g., Gaussian noise) to sensitive trajectories in such a way that it is provably unlikely for an adversary to infer the privatized trajectory. 
%    Formally, each agent’s state trajectory will
%    be made approximately indistinguishable from other nearby
%    state trajectories which the same agent individually could
%    have produced.

\begin{figure}[!t]
	\centering
	\includegraphics[scale=0.2]{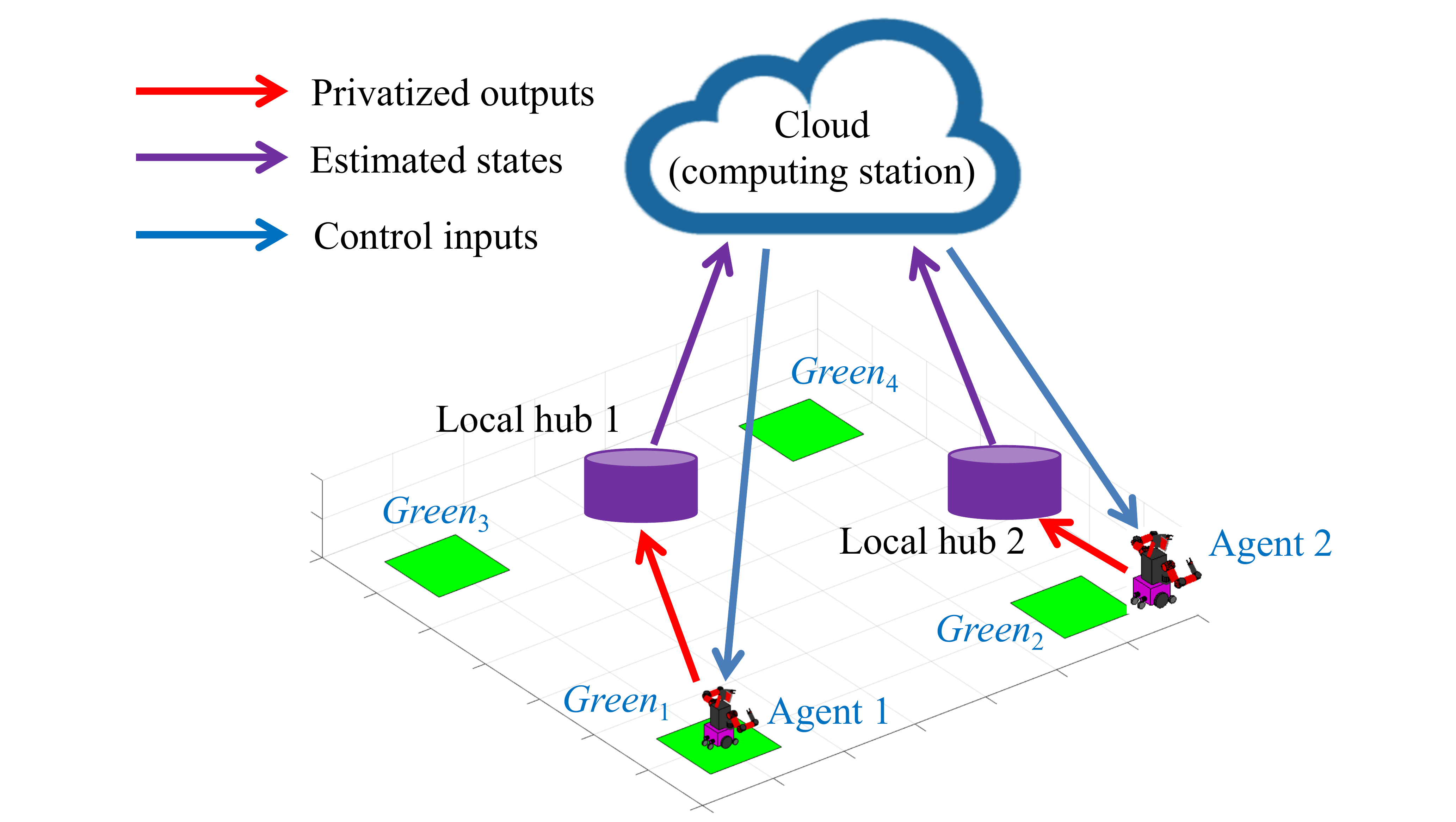}
	\caption{Diagram of information transmission among the agents, local hubs and the cloud.} 
	\label{fig_DF}
\end{figure}

	In this paper, we propose a controller synthesis approach for MAS that combines differential privacy with \textit{cloud-based} control subject to high-level specifications. In Fig.~\ref{fig_DF}, we provide an illustration of the setting that we consider in this paper. The two agents are supply robots transporting goods to warehouses which are marked by the green regions. Each agent reports its location with added privacy noise to their corresponding \textit{local hub}. Each local hub runs a Kalman filter to estimate the corresponding agent's location and periodically transmits the estimate to a cloud for decision-making. The cloud computes control inputs based on the state estimates and sends the control inputs to the agents.

    We express the high-level task specifications in \textit{metric temporal logic} (MTL), which has been used in many robotic applications \cite{Zhou2015OptimalMP, zhe_ijcai2019}. As in the  example in Fig. \ref{fig_DF},  MTL can express task specifications such as \textit{``Agent 1 should reach both $Green_1$ and $Green_4$ at least once in every consecutive 10 time units. Agent 2 should reach both $Green_2$ and $Green_3$ at least once in every consecutive 10 time units. The two agents should never collide with each other''}. 
    
    % \textcolor{blue}{[Why? Maybe the cloud knows what the agents are and synthesized a policy a priori? If we can do that, why share states?]}
    
    We model the dynamics of each agent as a stochastic control system and, it also has a \textit{nominal} deterministic control system. The Kalman filter in the local hubs can estimate the states of the stochastic control system and probabilistic bounds 
    on the states at the time instants when the cloud receives data from the local hubs. The cloud synthesizes the control inputs such that the trajectories of the \textit{nominal} deterministic control system satisfy the MTL specification with certain robustness margins. Then, utilizing a \textit{stochastic control bisimulation function} \cite{zheACC2018}, one can bound the
	divergence of the trajectories of a stochastic control system and its \textit{nominal} deterministic control system in a probabilistic
	fashion. In this way, the cloud can apply the synthesized control inputs (from the nominal deterministic control system) to the stochastic control system with a probabilistic guarantee for satisfying the MTL specification. 

	We provide an implementation of the proposed method on a simulation case study with two Baxter-On-Wheels robots as the agents. The results show that the synthesized controller can lead to satisfaction of the MTL specifications with a probabilistic guarantee.

\section{Preliminaries}  

\subsection{Stochastic Control Bisimulation Function}	
We consider a multi-agent system (MAS) consisting of $N$ agents. Let the time set be $\mathbb{T} = \mathbb{R}_{\ge0}$. For each agent $i$, we consider the stochastic control system with linear dynamics as below:
\begin{align}
\begin{split}
dx^i=&F^i(x^i,u^i)dt+G^i(x^i,u^i)dw
\\&=(A^ix^i+B^iu^i)dt+\Upsilon^i dw, 
\end{split}
\label{syslinear}
\end{align}
where $x^i:\mathbb{T}\rightarrow\mathcal{X}^i\subset\mathbb{R}^{n^i}$ ($n^i\in\mathbb{Z}_{>0}$) and $u^i:\mathbb{T}\rightarrow\mathbb{R}^{m^i}$ ($m^i\in\mathbb{Z}_{>0}$) are the state and input of the stochastic control system for agent $i$, $w$ is an $\mathbb{R}^{p^i}$-valued standard Brownian motion ($p^i\in\mathbb{Z}_{>0}$), and $A^i\in\mathbb{R}^{n^i\times n^i}$, $B^i\in\mathbb{R}^{n^i\times m^i}$ and $\Upsilon^i\in\mathbb{R}^{n^i\times p^i}$ are constant matrices.

We also consider the \textit{nominal control system} of (\ref{syslinear}) as the diffusionless deterministic version:
\begin{align}
&dx^{\ast i}=F(x^{\ast i},u^i)dt=(A^ix^{\ast i}+B^iu^i)dt,
\label{nom}
\end{align}
where $x^{\ast i}:\mathbb{T}\rightarrow\mathcal{X}^i\subset\mathbb{R}^{n^i}$ is the state of the nominal control system for agent $i$.

To bound the
divergence of the trajectories of a stochastic control system and its nominal control system, the stochastic control bisimulation function is introduced in \cite{zheACC2018}.

\begin{dfn}	\label{dfn:bisimulation}
	A twice differentiable function $\phi^i$ : $\mathcal{X}^i\times\mathcal{X}^i\rightarrow \mathbb{R}_{\geqslant 0}$ is a \textit{stochastic control bisimulation function} between (\ref{syslinear}) and its nominal system (\ref{nom}) if it satisfies
	\begin{align}
	\begin{split}
	&\phi^i(x^i,\tilde{x}^i)>0, \forall x^i,\tilde{x}^i\in \mathcal{X}^i,  x^i\neq\tilde{x}^i, \\
	&\phi^i(x^i, x^i)=0,\ \forall x^i\in \mathcal{X}^i,  
	\end{split}
	\end{align}
	and there exist $\mu^i, \alpha^i>0$ and a control input
	$u^i:\mathbb{T}\rightarrow\mathbb{R}^{m^i}$ such that 
	\begin{align}
	\begin{split}
	&\frac{\partial\phi^i}{\partial x^i}F^i(x^i,u^i)+\frac{\partial\phi^i}{\partial \tilde{x}^i}F^i(\tilde{x}^i,u^i)\\&+\frac{1}{2}{G^{i}}^T(x^i,u^i)\frac{\partial^{2}\phi^i}{\partial x^{i2}}G^i(x^i,u^i)\leq-\mu^i\phi^i+\alpha^i, 
	\end{split}
	\label{bisimulation_ineq}
	\end{align}
	for any $x^i, \tilde{x}^i\in\mathcal{X}^i$. 
\end{dfn}	

%\textcolor{blue}{[Can we add a sentence on what this definition says? Also, can we mention what $x$ and $x^*$ and $\tilde{x}$ each represent?]}

If the system is stable, i.e. $A^i$ is Hurwitz, we can construct a stochastic control bisimulation function of the form
\begin{center}
	$\phi^i(x^i, \tilde{x}^i)=(x^i-\tilde{x}^i)^{T}M^i(x^i-\tilde{x}^i)$, 
\end{center}
where $M^i$ is a symmetric positive definite matrix.

Based on \cite{zheACC2018}, if we pick $\alpha^i=\textnormal{tr}({\Upsilon^{i}}^T M^i\Upsilon^i)$, the inequality (\ref{bisimulation_ineq}) becomes a linear matrix inequality (LMI)
\begin{align}
& {A^{i}}^T M^i+M^iA^i+\mu^i M^i\preceq 0. 
\label{LMI} 
\end{align}
We denote agent $i$'s trajectory starting from $x^i_0$ with the input signal $u^i(\cdot)$ as $\xi^i_{\bm\cdot;x^i_0,u^i}$. Equation (\ref{bisimulation_ineq}) holds for any input signal $u^i(\cdot)$, so $u^i(\cdot)$ is free to be designed. It can also be seen that the matrix $M^i$ that satisfies (\ref{LMI}) also satisfies:
\begin{align}
&{A^{i}}^TM^i+M^iA^i\preceq 0. 
\label{LMI2} 
\end{align}
Thus it can be verified that $\psi(x^i, \tilde{x}^i)=\phi^i(x^i, \tilde{x}^i)=(x^i-\tilde{x}^i)^{T}M^i(x^i-\tilde{x}^i)$ is also a control bisimulation function (see Definition 2 of \cite{zhe_control}) of
the nominal system  
\begin{align}
\begin{split}
&dx^{\ast i}=(A^ix^{\ast i}+B^iu^i)dt. 
\end{split}
\label{nomlinear}
\end{align}
We denote the nominal system trajectory starting from $x^i_0$ with the input signal $u^i(\bm\cdot)$ as $\xi^{\ast i}_{{\bm\cdot};x^i_0,u^i}$. 

\begin{remark}	
	If the system dynamics is not stable but stabilizable, we can introduce another input signal $\zeta^i(\cdot)$ such that $u^i=K^ix^i+\zeta^i$,  where $K^i$ is chosen such that $A^i+B^iK^i$ is Hurwitz, and the above properties still hold by replacing (\ref{LMI}) with $(A^i+B^iK^i)^{T}M^i+M^i(A^i+B^iK^i)+\mu^i M^i\preceq 0$.
\end{remark}
%\textcolor{blue}{should this be $\zeta^i(x^i)$?}
\begin{proposition}
	If $\phi$ is a stochastic control bisimulation function between the stochastic system (\ref{syslinear}) and its nominal system (\ref{nomlinear}), then for any $t>0$ and $\eta^i\in[0,1)$,
	\begin{align}
	& P\left\lbrace \sup_{0\leq t'\leq t}\phi(\xi^{\ast i}_{t';x^i_0,u^i}, \xi^i_{t';x^i_0,u^i})<\frac{\alpha t}{1-\eta^i}\right\rbrace > \eta^i.  
	\label{prob}
	\end{align}
\end{proposition}

%\textcolor{blue}{[more explanation here for proof, also is $ \eta ^i $ constant?]}

\begin{proof} 
	Straightforward from Proposition 2.2 of \cite{Julius2008CDC} and (\ref{bisimulation_ineq}).
\end{proof}  	
 In Equation (\ref{prob}), $\phi$ provides a probabilistic upper bound for the distance between the states of the stochastic system and its nominal system in a finite time horizon.

\subsection{Metric Temporal Logic (MTL)}        
\label{MTL}   
In this subsection, we briefly review  metric temporal logic (MTL)~\cite{FainekosMTL}. The state $x$ of the system belongs to
the domain $\mathcal{X}\subset\mathbb{R}^n$. The time set is $\mathbb{T} = \mathbb{R}_{\ge0}$. The domain $\mathbb{B} = \{\textrm{True}, \textrm{False}\}$ is the Boolean domain, and the time index set is $\mathbb{I} = \{0,1,\dots\}$. We use $t[k]\in\mathbb{T}$ to denote the time instant at time index $k\in\mathbb{I}$ and $x[k]\triangleq x(t[k])$ to denote the value of $x$ at time $t[k]$. With slight abuse of notation, we use $\xi$ to denote an trajectory of the system as a function from $\mathbb{T}$ to $\mathcal{X}$. A set $AP$ is a set of atomic propositions, each mapping $\mathcal{X}$ to $\mathbb{B}$. The syntax of MTL is defined recursively as follows:
\[
\varphi:=\top\mid \pi\mid\lnot\varphi\mid\varphi_{1}\wedge\varphi_{2}\mid\varphi_{1}\vee
\varphi_{2}\mid\varphi_{1}\mathcal{U}_{\mathcal{I}}\varphi_{2},
\]
where $\top$ stands for the Boolean constant True, $\pi\in AP$ is an atomic
proposition, $\lnot$ (negation), $\wedge$ (conjunction), $\vee$ (disjunction)
are standard Boolean connectives, $\mathcal{U}$ is a temporal operator
representing \textquotedblleft until\textquotedblright, $\mathcal{I}$ is a time index interval of
the form $\mathcal{I}=[i_{1},i_{2}]$ ($i_1\le i_2$, $i_1, i_2\in\mathbb{I}$). We
can also derive two useful temporal operators from \textquotedblleft
until\textquotedblright~($\mathcal{U}$), which are \textquotedblleft
eventually\textquotedblright~$\Diamond_{\mathcal{I}}\varphi=\top\mathcal{U}_{\mathcal{I}}\varphi$ and
\textquotedblleft always\textquotedblright~$\Box_{\mathcal{I}}\varphi=\lnot\Diamond_{\mathcal{I}}\lnot\varphi$. We define the set of states that satisfy the atomic proposition $\pi$ as $\mathcal{O}(\pi)\subset \mathcal{X}$. 
%\textcolor{blue}{should this be $\mathcal{O}(\pi)\subset \mathcal{X}$ instead?} 

We denote the distance from $x$ to a set $D\subseteq\mathcal{X}$ as \textbf{dist}$_d(x,D)\triangleq$inf$\{d(x, x')\mid x'\in cl(D)\}$, where $d$ is a metric on $\mathcal{X}$ and $cl(D)$ denotes the closure of the set $D$. In this paper, we use the metric $d(x,x')=\norm{x-x'}$, where $\left\Vert\cdot\right\Vert $ denotes the 2-norm. We denote the depth of $x$ in $D$ as \textbf{depth}$_d(x,D)\triangleq$ \textbf{dist}$_d(x,\mathcal{X}\setminus D)$. We define the signed distance from $x$ to $D$ as $\textbf{Dist}_d(x,D)\triangleq-\textbf{dist}_d(x,D)$, if $x$ $\not\in D$; and $\textbf{Dist}_d(x,D)\triangleq\textbf{depth}_d(x,D)$, if $x$ $\in D$.
%\begin{equation}
%
%\begin{cases}
%-\textbf{dist}_d(x,D),& \mbox{if $x$ $\not\in D$},\\                      
%\textbf{depth}_d(x,D), & \mbox{if $x$ $\in D$}.
%\end{cases}                        
%\end{equation}

The Boolean semantics of MTL can be found in \cite{FAINEKOScontinous}, with the slight variation that we only evaluate the satisfaction of a trajectory with respect to an MTL formula at discrete-time instants $t[k]~(k\in\mathcal{I})$. The robustness degree of a trajectory $\xi$ with respect to an MTL formula $\varphi$ at time index $k$, denoted as $\left[\left[\varphi\right]\right](\xi, k)$, is defined recursively as follows:
\[
\begin{split}
\left[\left[\top\right]\right](\xi, k) :=& +\infty,\\
\left[\left[ \pi\right]\right](\xi, k)  :=&\textbf{Dist$_d(x[k],\mathcal{O}(\pi))$},\\
\left[\left[ \neg\varphi\right]\right](\xi, k)  :=&-\left[\left[ \varphi\right]\right](\xi, k),\\
\left[\left[\varphi_1\vee\varphi_2\right]\right](\xi, k)  :=&\max\big(\left[\left[ \varphi_1\right]\right](\xi, k),\left[\left[ \varphi_2\right]\right](\xi, k)\big),\\
\left[\left[\varphi_1\mathcal{U}_{\mathcal{I}}\varphi_{2}\right]\right](\xi, k)  :=&\max_{k'\in (k+\mathcal{I})}\Big(\min\big(\left[\left[ \varphi_2\right]\right](\xi, k'),\\& \min_{k\leq k''<k'}\left[\left[\varphi_1\right]\right]
(\xi,k'')\big)\Big).  
\end{split} 
\]
As defined, $\left[\left[\varphi\right]\right](\xi, k)\ge0$ if $\xi$ satisfies $\varphi$ at time index $k$. 

\subsection{Differential Privacy for Dynamical Systems}	

In this subsection, we review the theoretical framework of differential privacy for dynamical systems.
%Denote agent $i$'s state trajectory by $\xi^i_{\bm\cdot;x^i_0,u^i}$. The $k^{th}$ element of $\xi^i_{\bm\cdot;x^i_0,u^i}$ is denoted by $x^{i}[k]\in\mathcal{X}^i\subset\mathbb{R}^{n_{i}}$
%for some $n_{i}\in\mathbb{N}$, and we also set ${ n =\sum_{i=1}^{N}n_i }$, which is the dimension of all agents' states in aggregate. 
The notion of differential privacy in this paper follows the definition of differential privacy for trajectories introduced in \cite{LeNy2014} and applied in \cite{YazdaniLQ2018}. 
%Differential privacy can be used to ensure that an adversary is unlikely to determine either the input or state trajectory of a system, and in this paper we implement differential privacy to protect state trajectories.

We consider the so-called ``input perturbation'' approach to differential privacy. This means that each agent will directly add noise to its own outputs before sharing them with a local hub. This has the advantage of \textit{privatizing} sensitive data before it is shared. Formally, each agent's state trajectory will be made approximately indistinguishable from other nearby state trajectories which the same agent individually could have produced; the notions of ``nearby'' and ``approximately indistinguishable'' are formalized below in Definitions~\ref{def:adj} and~\ref{dfn:differential privacy}. 
% Privacy guarantees are likewise provided on an individual basis, and the input perturbation approach will be used below to protect individual agent's state trajectories. 

With slight abuse of notations, we consider discrete-time trajectories 
$\xi^{i}=[x^i[0],$ $x^i[1],\dots,x^i[k],\dots]$,
where $x^i[k]=\xi^i_{t[k];x^i_0,u^i}\in\mathbb{R}^{n_i}$ for all $k$. We also
use the~$\ell_p$-norm
${
	\left\Vert \xi^{i}\right\Vert _{\ell_{p}}:=\left(\sum_{k=1}^{\infty}\left\Vert x^i[k]\right\Vert _{p}^{p}\right)^{\frac{1}{p}},
}$
where $\left\Vert .\right\Vert _{p}$ is the ordinary $p\text{-norm}$
on $\mathbb{R}^{d}$. We further define the set 
$
\ell_{p}^{d}:=\left\{ \xi^{i}\mid x^i[k]\in\mathbb{R}^{d},\;\left\Vert \xi^{i}\right\Vert _{\ell_{p}}<\infty\right\} .
$

The state trajectory $\xi^i$ is contained
in the set $\tilde{\ell}_{2}^{n_i}$, which is the set of sequences
of vectors in $\mathbb{R}^{n_i}$ whose finite truncations are all
in $\ell_{2}^{n_i}$. Formally, we define the truncation operator
$P_{L}$ over trajectories as follows: $P_{L}\left(\xi^i\right)=x^i[k]$, if $k\le L$; and $P_{L}\left(\xi^i\right)=0$, otherwise.
We say that $\xi^i \in\tilde{\ell}_{2}^{n_i}$ if and only if      
$P_{L}[\xi^i ]\in\ell_{2}^{n_i}$ for all $L\in\mathbb{N}$.

A differentially private mechanism makes adjacent trajectories
produce outputs which are similar in a precise sense,
making the state trajectories approximately indistinguishable to the recipient of a system's outputs.
% The choice of adjacency relation is a key
%part of any differential privacy implementation because it specifies which sensitive pieces of data must be
%made approximately indistinguishable.  
To formulate differential privacy for trajectories, we next define the adjacency relation
over the space $\tilde{\ell}_{2}^{n_i}$
defined above. 
\begin{dfn} \label{def:adj} \emph{(Adjacency) }
	Fix an adjacency parameter ${\nu^i > 0}$ for agent $ i $. 
	The adjacency relation $\adj_{\nu^i}$ is defined
	for all $\xi^i, \xi'^i \in \tell^{n_i}_2$ as %\textcolor{blue}{ZX: change v and w to $\xi^i$, $\xi'^i$?}
	\begin{equation}
	\adj_{\nu^{i}}(\xi^i, \xi'^i) = 
	\begin{cases}
	1, & \mbox{if}~ \|\xi^i - \xi'^i\|_{\ell_2} \leq \nu^i, \\
	0, & \text{otherwise}. 	
	\end{cases}
	\end{equation}
\end{dfn}

Two state trajectories of agent $ i $ are thus adjacent if the~$\ell_2$ distance between them
is not more than~$\nu^i$. Differential privacy must therefore make agent $ i $'s
state trajectory approximately indistinguishable from all others contained
in an~$\ell_2$-ball of radius~$\nu^i$ centered on its actual trajectory.

Next is a formal definition of differential privacy for dynamical systems which specifies the probabilistic guarantees of privacy. 
To state it, we will use a probability space $\text{(\ensuremath{\Omega}, \ensuremath{\mathcal{F}}, \ensuremath{\mathbb{P}}})$.
This definition considers outputs in the space $\tilde{\ell}_{2}^{q_i}$
and uses a $\sigma\text{-algebra}$ over $\tilde{\ell}_{2}^{q_i}$,
denoted $\Theta_{2}^{q_i}$, construction of which can be found in~\cite{hajek2015random}. 
\begin{dfn}\label{dfn:differential privacy}
	\emph{($(\epsilon^i,\delta^i)$-Differential Privacy for Agent $i$)} 
	With\emph{ $\epsilon^i>0$} and $\delta^i\in\left(0,\nicefrac{1}{2}\right)$ for agent $i$, a mechanism       
	$\mathcal{M}:\tilde{\ell}_{2}^{n_i}\times\Omega\rightarrow\tilde{\ell}_{2}^{q_i}$
	is $\left(\epsilon^i,\delta^i\right)$-differentially private if for all adjacent ${ \xi^i ,{\xi'^i}\in\tilde{\ell}_{2}^{n_i} }$ and for all $S\in\Theta_{2}^{q_i}$,
	we have
	\begin{equation}
	\mathbb{P}\left[\mathcal{M}\left(\xi^i\right)\in S\right]\le e^{\epsilon^i}\mathbb{P}\left[\mathcal{M}\left({\xi'^i}\right)\in S\right]+\delta^i.
	\end{equation}
\end{dfn}

At time~$k$, agent~$i$ has state~$x^i(k)\in \mathbb{R}^{n_i}~(n_i\in\mathbb{N})$. We convert the continuous-time linear dynamics in (\ref{syslinear}) to the following discrete-time linear dynamics for agent $i$
\begin{equation}
\begin{split}
x^i[k+1] &= \bar{A}^i x^i[k]+\bar{B}^i u^i[k]+w^i[k],\\
y^i[k] &= \bar{C}^i x^i[k],
\end{split}
\label{eq:dynamics}
\end{equation}
where $u^i[k]\in \mathbb{R}^{m_i}$ is the input, process noise for agent $i$ is denoted by $w^i[k]\in\mathbb{R}^{n_i}$, and the matrices $\bar{A}^i\in\mathbb{R}^{n_i\times n_i}$, $\bar{B}^i\in\mathbb{R}^{n_i\times m_i}$ are derived from $A^i$ and $B^i$ in (\ref{syslinear}) for the discrete-time
state-space representation of agent $i$, and $ \bar{C}^i\in \mathbb{R}^{q_i\times n_i} $. The probability distribution of the process noise is given by~$w^i[k]\sim\mathcal{N}\left(0,W^i\right)$, 
where $0\prec W^i\in\mathbb{R}^{n_i\times n_i}$, and all process noise terms are assumed to have finite variance.               

At each time $k$, agent $ i $ outputs the value $y^i[k]$ and we define $\rho ^i = [y^i[0],y^i[1],\dots,y^i[k],\dots]$. 
Absent any privacy protections, the values of~$y^i$ could reveal those of~$x^i$ over time, which would compromise agent~$i$'s privacy by revealing its state trajectory. Therefore, noise must be added to agent~$i$'s output to protect its state trajectory. Calibrating the level of noise is done using the ``sensitivity'' of an agent's output, which we define next for the input perturbation privacy.
%
% that we use; we emphasize that, although agents perturb the outputs of their own dynamics, the term ``input
%perturbation" label applies because agents perturb what will become the \emph{inputs} to a Kalman filter. 
%\textcolor{blue}{ZX: This above sentence is very hard to understand. What is "label"?}
%The following bound is adapted from~\cite[Section IV-A]{LeNy2014}.

\begin{dfn}\label{dfn:sensitivity for dynamic systems}
	\emph{(Sensitivity for Input Perturbation Privacy)} 
	The $\ell_{2}\text{-norm}$ sensitivity of
	agent~$i$'s output map is the greatest distance
	between two output trajectories which correspond to adjacent state
	trajectories. Formally, for $\xi^i,{\xi'^i}\in\tilde{\ell}_2^{n_i}$,
	\begin{equation}
	\Delta_{\ell_{2}}\rho^i: =\underset{\xi_i,\xi'_i\mid\textnormal{Adj}_{\nu^i}(\xi^i,{\xi'^i})=1}{\sup}\left\Vert \bar{C}^i \xi^i-\bar{C}^i{\xi'^i}\right\Vert _{\ell_{2}}. 
	\end{equation}
\end{dfn}
We can bound $\Delta_{\ell_{2}}\rho^i$ via $ \Delta_{\ell_{2}}\rho^i \le \norm{\bar{C}^i} \nu^i $
% \textcolor{blue}{ZX: what is $s_1$?}
 \cite{LeNy2014}, where $\norm{\bar{C}^i}$ denotes the largest singular value of $\bar{C}^i$.
Various mechanisms have been developed for enforcing differential privacy in the literature \cite{Dwork2014algorithmic}. The Gaussian mechanism requires adding Gaussian noise to outputs to mask agents' state trajectories, and it can be useful in control settings that are robust to Gaussian noise. We next provide a definition of the Gaussian mechanism in terms of the $ \mathcal{Q} $-function, defined by $
\mathcal{Q}\left(y\right)=\frac{1}{\sqrt{2\pi}}\int_{y}^{\infty}e^{-\frac{z^{2}}{2}}dz
$.

\begin{lemma}\label{lem:gaussian mechanism}
	\emph{(Input Perturbation Gaussian Mechanism for Linear Systems)}
	Let agent~$i$ specify privacy parameters
	$\epsilon^i>0$ and $\delta^i\in\left(0,\nicefrac{1}{2}\right)$. 
	Let~$\rho^i \in \tilde{\ell}_2^{q_i}$ denote the output of a system with state trajectories
	in~$\tilde{\ell}_{2}^{n_i}$,
	and denote its ${\ell}_{2}$-norm sensitivity by $\Delta_{\ell_{2}}\rho^i$.
	Then the Gaussian mechanism for $\left(\epsilon^i,\delta^i\right)$-differential privacy
	takes the form
	\begin{equation}
	\tilde{y}^i[k]=y^i[k]+v^i[k],
	\end{equation}
	where~$v^i$ is a stochastic process with $v^i[k]\sim\mathcal{N}\left(0,{\sigma^i}^{2}I_{q_i}\right)$,
	$I_{q_i}$ is the $q_i\times q_i$ identity matrix, and 
	\begin{equation}\label{eq:sigma}
	\sigma^i\ge\frac{\Delta_{\ell_{2}}\rho^i}{2\epsilon^i}\left(\iota_{\delta_i}+\sqrt{\iota_{\delta^i}^{2}+2\epsilon^i}\right)\text{ where }\iota_{\delta^i}:=\mathcal{Q}^{-1}\left(\delta^i\right). 
	\end{equation}
	This Gaussian mechanism provides $\left(\epsilon^i,\delta^i\right)$-differential privacy.
\end{lemma}

\emph{Proof:} See \cite[Corollary 1]{LeNy2014}. \hfill $\blacksquare$                   

\noindent In words, the Gaussian mechanism adds i.i.d Gaussian
noise point-wise in time to the output of a system to keep its state trajectory private. We will use the Gaussian mechanism to enforce differential privacy for the remainder of the paper.

\section{Differentially Private Controller Synthesis with Metric Temporal Logic Specifications}
In this section, we first present the problem formulation of differentially private controller synthesis with metric temporal logic specifications, then provide the theoretical framework and algorithm for solving the problem.

\begin{subsection}{Problem Formulation}\label{subsec:problemForm}

To formulate the problem, we first define the network-level dynamics. 
We consider the stochastic control system with the aggregated states as below:
\begin{align}
\begin{split}
&dx=(Ax+Bu)dt+\Upsilon dw, 
\end{split}
\label{syslinear_aggregate}
\end{align}
where $x=[(x^{1})^T,\dots,(x^{N})^T]^T $ and $u=[(u^{1})^T,$ $\dots,(u^{N})^T]^T$, where $A=\textnormal{diag}(A_1,\dots,A_N)$ is a block diagonal matrix with blocks $A_1$ through $A_N$, $B=\textnormal{diag}(B_1,\dots,B_N)$, $C=\textnormal{diag}(C_1,\dots,C_N)$, and $ \Upsilon = \textnormal{diag}\left[ \Upsilon^1,\dots,\Upsilon^N \right] $.  We denote the aggregated system trajectory starting from $x_0=[(x_0^{1})^T,\dots,(x_0^{N})^T]^T$ with the input signal $u(\cdot)$ as $\xi_{\bm\cdot;x_0,u}$. 

Each agent reports its state information with added privacy noise (e.g., Gaussian noise) to a \textit{local hub}. Each local hub runs a Kalman filter to estimate the state of the corresponding agent and periodically (with period $T\in\mathbb{Z}_{>0}$) send the state estimate of the agent to a \textit{cloud}. The cloud computes the optimal inputs for a control horizon of $T$ time instants for each agent with respect to an MTL specification $\varphi$. We formulate the controller synthesis problem as follows.
 
%\textcolor{blue}{[is there one local hub for every agent? Kasra: Explain H and N kalman filter]}

%\textcolor{blue}{[why $ \varphi $ is only known by cloud? Kasra: why $ \varphi $ is okay to be public?]}

%\textcolor{blue}{[what does this mean? $\mathbb{P}\{\left[\left[\varphi\right]\right](\xi_{\bm\cdot;x_0,u}, 0)\ge0\}\ge \kappa$ explain]}

\begin{problem} 
	Given an MAS consisting of $N$ agents, $N$ local hubs, a cloud, an MTL specification $\varphi$ and privacy parameters $\epsilon^i\in[\epsilon_{\min}, \epsilon_{\max}]$ (where $0<\epsilon_{\min}\le\epsilon_{\max}$), $\delta^i\in[\delta_{\min}, \delta_{\max}]$ (where $0<\delta_{\min}\le\delta_{\max}<\nicefrac{1}{2}$), compute the input signals $u^1(\bm\cdot), \dots, u^T(\bm\cdot)$ that minimize $\sum_{i=1}^{N}\norm{u^i(\bm\cdot)}$ while satisfying $\mathbb{P}\{\left[\left[\varphi\right]\right](\xi_{\bm\cdot;x_0,u}, 0)\ge0\}\ge \chi$ for given $\chi\in(0,1]$, i.e., the trajectory $\xi_{\bm\cdot;x_0,u}$ satisfies the MTL specification $\varphi$ with probability at least $\chi$. 
	\label{problem}                               
\end{problem}        

\end{subsection}

\begin{subsection}{Kalman Filtering at Local Hubs}

We assume that the privatized output of each agent is transmitted to the local hubs in discrete-time where the discrete-time dynamics for agent $i$ is as follows:
\begin{equation}
\begin{split}
x^i[k+1] &= \bar{A}^i x^i[k]+\bar{B}^i u^i[k]+w^i[k],\\
\tilde{y}^i[k] &= \bar{C}^i x^i[k]+v^i[k],
\end{split}
\label{eq:dynamics}
\end{equation}  
where the privacy noise $ v^i[k]\sim\mathcal{N}(0,{\sigma^i}^2I_{q_i})$ is a Gaussian random variable and where $\sigma^i$ is chosen according to Equation~\eqref{eq:sigma} corresponding to privacy parameters $ (\epsilon ^i, \delta ^ i) $. We also define the privacy covariance matrix $ V^i := {\sigma^i}^2I_{q_i} $.

%\begin{equation}
%\begin{split}
%x^i[k+1] &= \bar{A}^i x^i[k]+\bar{B}^i u^i[k]+w^i[k]\\
%y^i[k] &= \bar{C}^i x^i[k],
%\end{split}
%\label{eq:networkDynamics}
%\end{equation}
%where $ x[k]=[(x^1[k])^T,\dots,(x^N[k])^T]^T $, $ u[k]=[(u^1[k])^T,\dots,(u^N[k])^T]^T $, $ y[k]=[(y^1[k])^T,\dots,(y^N[k])^T]^T $, $\bar{A}= \textnormal{diag}(\bar{A}_1,\dots,\bar{A}_N)$, $\bar{B}= \textnormal{diag}(\bar{B}_1,\dots,\bar{B}_N)$, and $\bar{C}= \textnormal{diag}(\bar{C}_1,\dots,\bar{C}_N)$.                                              
The local hubs are responsible for estimating the agents' states. For agent $i$, the prediction step of the Kalman filter is given by
\begin{equation}\label{kalmanfilterprior} 
\hat{x}^{-i}[k+1]=\bar{A}^i\hat{x}^i[k]+\bar{B}^iu^i[k],
\end{equation}                                                  
where $ \hat{x}^{-i}[k] $ is the \emph{a priori} state estimate. The \emph{a posteriori} state estimate $ \hat{x}^i[k] $ is updated as
\begin{equation}\label{eq:kalmanfilter}
\hat{x}^i[k+1]= \hat{x}^{-i}[k+1]+\barsig ^i \bar{C^i}^{T}{V^i}^{-1}(\tilde{y}^i[k+1]-\bar{C}^i \hat{x}^{-i}[k+1]).
\end{equation}
                                                                     
The \emph{a posteriori} error covariance matrix $\barsig^i$ is given by
\begin{align}\label{eq:Sigmabar}
\barsig^i&=\Sigma^i-\Sigma^i \bar{C^i}^{T}\left(\bar{C^i}\Sigma^i \bar{C^i}^{T}+V^i\right)^{-1}\bar{C^i}\Sigma^i
\end{align}
where the \emph{a priori} error covariance matrix $\Sigma^i$ is the unique positive semidefinite solution
to the discrete algebraic Riccati equation
\begin{multline}\label{eq:DARE}
\Sigma^i=\bar{A^i}\Sigma^i \bar{A^i}^{T}-\bar{A^i}\Sigma^i \bar{C^i}^{T}\left(\bar{C^i}\Sigma^i \bar{C^i}^{T}+V^i\right)^{-1}\bar{C^i}\Sigma^i \bar{A^i}^{T} \\ +W^i 
\end{multline}                            
Increasing the level of privacy is achieved by adding more noise, namely, using smaller privacy parameters $ \epsilon ^i $ and $\delta ^i$  translates to imposing a larger $ \sigma ^ i$. Larger $ \sigma ^i $ accordingly makes the estimation more uncertain, hence making the privacy covariance $ V^i $ larger.  As seen in Equation~\eqref{eq:DARE}, increasing the noise naturally results in a larger \emph{a priori} covariance matrix $ \Sigma ^i $ and \emph{a posteriori} covariance matrix $ {\overline{\Sigma}}^i $ in Equation~\eqref{eq:Sigmabar}. Therefore, as we increase the strength of privacy, i.e. decrease the privacy values $ \epsilon ^i$ and $ \delta ^i $, the covariance matrices $ \Sigma ^i $ and $ {\overline{\Sigma}}^i $ monotonically get larger. 

As we did in Subsection~\ref{subsec:problemForm}, we assemble the covariance matrices $\Sigma^i $ and ${\overline{\Sigma}}^i$ into network-level matrices $ \Sigma = \textnormal{diag}\left(\Sigma^i,\dots,\Sigma^N \right) $ and $ \overline{\Sigma} = \textnormal{diag}\left( {\overline{\Sigma}}^1,\dots,{\overline{\Sigma}}^N \right) $. We also assume $\epsilon ^i \in[\epsilon_{\min}, \epsilon_{\max}]$ and $\delta^ i\in[\delta_{\min}, \delta_{\max}]$. The mean squared error (MSE) of the estimated states by the Kalman filter is computed by 
\begin{equation}\label{eq:mse}
\text{MSE} = \mathbb{E}[\Vert x[k]-\hat{x}[k]\Vert^{2}\big]=\text{tr}\overline{\Sigma}
\end{equation}
and we denote the upper and lower bounds on the MSE by 
\begin{equation}\label{key}
\left(\text{tr}\overline{\Sigma} \right)_{\min}\le\textnormal{MSE}\le \left(\text{tr}\overline{\Sigma} \right)_{\max},
\end{equation}
where $ \left(\text{tr}\overline{\Sigma} \right) _{\min} $ corresponds to the lowest level of privacy achieved by $(\epsilon_{\max} ,\delta_{\max})$ and $ \left(\text{tr}\overline{\Sigma} \right) _{\max} $, corresponds to the highest level of privacy, achieved by $(\epsilon_{\min} ,\delta_{\min})$. The following two lemmas bound the MSE values of interest.
\begin{lemma}\label{lem:TraceofProduct}
	Let $R$ and $M$ be $n\times n$ matrices. If $R=R^{T}\succeq0$
	and $M$ is symmetric, then 	
	\begin{equation}
	\lambda_{n}(M)\textnormal{tr}(R)\leq\textnormal{tr}(RM)\leq\lambda_{1}(M)\textnormal{tr}(R).
	\end{equation}
\end{lemma}
\begin{IEEEproof}
	See \cite[Fact 5.12.4]{Bernstein2009}. \hfill $\blacksquare$
\end{IEEEproof}

\begin{lemma}\label{lem:mseBound}
	Assume $M\succeq0 $. Then the MSE of estimation error scaled by the matrix $M$ is bounded by
	\begin{equation*}\label{key}
	\resizebox{.5 \textwidth}{!} 
	{
		$ \lambda_{n}(M)\left(\text{tr}\overline{\Sigma} \right)_{\min}\leq \mathbb{E}\left[\left(x-\hat{x}\right)^{T}M\left(x-\hat{x}\right)\right]\leq\lambda_{1}(M)\left(\text{tr}\overline{\Sigma} \right) _{\max}. $
	}	
	\end{equation*}
\end{lemma}

\begin{IEEEproof}
	By immediate expansion and using Equation~\eqref{eq:mse} we get 
	\begin{equation*}\label{key}
	\mathbb{E}\left[\left(x-\hat{x}\right)^{T}M\left(x-\hat{x}\right)\right]=\text{tr}(M\overline{\Sigma})
	\end{equation*}
	and using Lemma~\ref{lem:TraceofProduct} we find 
	\begin{equation}\label{key}
	\lambda_{n}(M)\textnormal{tr}(\overline{\Sigma})\leq\textnormal{tr}(M\overline{\Sigma})\leq\lambda_{1}(M)\textnormal{tr}(\overline{\Sigma}),
	\end{equation}
	which completes the proof.  \hfill $\blacksquare$
\end{IEEEproof}

Next, we use Markov's inequality to write
\begin{multline}\label{key}
\mathbb{P}\left[( x[k]-\hat{x}[k] )^TM( x[k]-\hat{x}[k] )\ge \beta\right]\\\le \frac{\mathbb{E}\left[ ( x[k]-\hat{x}[k] )^TM( x[k]-\hat{x}[k] ) \right]}{\beta},
\end{multline}
and using Lemma~\ref{lem:mseBound} we can say
\begin{equation}\label{key}
\mathbb{P}\left[( x[k]-\hat{x}[k] )^TM( x[k]-\hat{x}[k] )\ge \beta\right]\le \frac{\textnormal{tr}(M\overline{\Sigma})}{\beta}.
\end{equation}

\end{subsection}

\begin{subsection}{Controller Synthesis With MTL Specifications in Cloud}
 In this subsection, we consider the controller synthesis problem (in the cloud) with MTL specifications.
	
%	\textcolor{blue}{The $ \tau $ notation is weird because we used the superscripts already for agents $ i $}

 We first provide the following theorem that bounds the robustness degrees of trajectories of a stochastic control system and those of the nominal deterministic control system with respect to an MTL specification.
 
	\begin{theorem} 
	\label{gamma}
	For a group of $N$ agents, a cloud and any MTL specification $\varphi$, if for a given time $\tau$ the following holds
	\begin{align}                                                                    
	\begin{split}
	&\mathbb{P}\left[( x_{\tau}-\hat{x}_{\tau} )^TM( x_{\tau}-\hat{x}_{\tau} )<\beta\right]>\gamma,
	\end{split}
	\label{end}
	\end{align} 
	where $x_{\tau}\triangleq\xi_{\tau;x_0,u}$ and $\hat{x}_{\tau}$ respectively denote the true network-level state at time $\tau$ and the \emph{a posteriori} state estimate of the Kalman filter for the network-level state at time $\tau$, $t\ge\tau$, $\beta>0, \gamma=1-\frac{\textnormal{tr}(M\overline{\Upsilon})}{\beta}\in(0,1]$, $\eta\in[0,1)$
	then we have 
	\begin{align} \nonumber
	&\mathbb{P}\left\lbrace  \vert \left[\left[\varphi\right]\right](\xi_{\bm\cdot; x_{\tau}, u}, t) -\left[\left[\varphi\right]\right](\xi^{\ast}_{\bm\cdot; \hat{x}_{\tau}, u}, t)\vert<\hat{\beta}\right\rbrace  >\gamma\eta,
	\end{align}  
	where $\hat{\beta}\triangleq\Big(\sqrt{\beta}+\sqrt{\frac{\alpha (t-\tau)}{1-\eta}}\Big)\norm{M}^{-\frac{1}{2}}$, $\alpha=\textnormal{tr}({\Upsilon}^T M\Upsilon)$, and $\norm{M}$ denotes the largest singular value of the matrix $M$.                      	                                                      
	\end{theorem} 

%\textcolor{blue}{tell the reader what Equation 33 is}
%
%\textcolor{blue}{state what this Theorem is actually saying}

%\begin{proof} 
%	See Appendix.
%\end{proof} 

\begin{algorithm}[tp]
	\caption{Differentially private controller synthesis with MTL specifications.}                                                                   
	\label{MTLalg}
	\begin{algorithmic}[1]
		\State $\ell\gets0$, initialize the states and inputs
		\While{$\ell<T_{\textrm{max}}$}
		\If{$\ell=jT$ for some $j=0,1,\dots$} 
		\State Obtain $\hat{x}^i[\ell]$ from the Kalman filter in local hub
		\indent \quad $i$ and update the constraints (\ref{constraint_s})-(\ref{constraint_e})  \label{update}
		\State $\varphi\gets [\varphi]^{\ell}_{0}$ \label{update_MTL}
		\State Re-solve (\ref{obj})-(\ref{constraint_e}) to obtain the optimal inputs \indent \quad$u^{\ast i}[\ell+q]~(i=1,2,\dots, Q, q=0,1,\dots, T-1)$	\label{recompute}			
		\State $\hat{u}^{\ast i}[\ell+q]\gets u^{\ast i}[\ell+q]~(i=1,2,\dots, Q$,\\ \indent \quad\quad\quad $q=0,1,\dots, T-1)$         \label{replace}    
		\EndIf                  
		\EndWhile                    
		\State Return \(\hat{u}^{\ast}\)
	\end{algorithmic} 
\end{algorithm}
\end{subsection}

 Under the conditions of Theorem \ref{gamma}, if $\left[\left[\varphi\right]\right](\xi^{\ast}_{\bm\cdot; \hat{x}_{0}, u}, 0)\ge\hat{\beta}$, then $\mathbb{P}\left\lbrace   \left[\left[\varphi\right]\right](\xi_{\bm\cdot; x_{0}, u}, 0) >0\right\rbrace  >\gamma\eta$. Therefore, the cloud can synthesize the control inputs for the \textit{nominal} deterministic control system such that the trajectories of the \textit{nominal} deterministic control system satisfy the MTL specification with certain robustness margins. Then all the synthesized control input methods for the deterministic system can be applied to the
 stochastic control system with probability of at least $\gamma\eta$ for satisfying the MTL specification. 

We use $[\varphi]^{\ell}_{k}$ to denote the formula modified from the MTL formula $\varphi$ when $\varphi$ is evaluated at time index $k$ and the current time index is $\ell$. $[\varphi]^{\ell}_{k}$ can be calculated recursively as follows (we use $\pi_{k}$ to denote the atomic predicate $\pi$ evaluated at time index $k$):
\begin{align}
\begin{split}                     
[\pi]^{\ell}_{k} =&  
\begin{cases}
\pi_{k},& \mbox{if $k>\ell$}\\  	
\top,& \mbox{if $k\le \ell$ and $x[k]\in\mathcal{O}(\pi)$}\\  
\bot,& \mbox{if $k\le \ell$ and $x[k]\not\in\mathcal{O}(\pi)$}                                                 
\end{cases}\\
[\neg\varphi]^{\ell}_{k}  :=&\neg[\varphi]^{\ell}_{k}\\
[\varphi_1\wedge\varphi_2]^{\ell}_{k}:=&[\varphi_1]^{\ell}_{k}\wedge[\varphi_2]^{\ell}_{k}\\
[\varphi_1\mathcal{U}_{\mathcal{I}}\varphi_2]^{\ell}_{k} :=&\bigvee_{k'\in (k+\mathcal{I})}\Big([\varphi_2]^{\ell}_{k'}\wedge\bigwedge_{k\le k''<k'}[\varphi_1]^{\ell}_{k''}\Big).
\end{split}
\label{update_phi}
\end{align}

If the MTL formula $\varphi$ is evaluated at the initial time index (which is the usual case when the task starts at the initial time), then the modified formula is $[\varphi]^{\ell}_{0}$. 

For example, if $\varphi=\Box_{[0,10]}(x\ge5)$, the current time is 5 and $\varphi$ is not violated yet, then $[\varphi]^{5}_{0}=\Box_{[6,10]}(x\ge5)$. 

Algorithm 1 shows the proposed differentially private controller synthesis approach with respect to MTL specifications. 
The controller synthesis problem can be formulated as a sequence of mixed integer linear programming problems:
\begin{align}
\underset{u^i[\ell:\ell+T-1]}{\argmin} ~ & \sum_{i=1}^{N}\norm{u^i[\ell:\ell+T-1]}  \label{obj} \\
\text{subject to:} ~ 
& x^{\ast i}[\ell+k+1]=\bar{A}^ix^{\ast i}[\ell+k]+\bar{B}^iu^i[\ell+k], \nonumber \\ & x^{\ast i}[\ell]= \hat{x}^{i}[\ell], \forall i=1,\dots,N, ~\forall k=0,\dots,T-1,  \label{constraint_s}  \\
& u^i_{\textrm{min}}\le u^i[\ell+k]\le u^i_{\textrm{max}}, \forall i=1,2,\dots,N, \nonumber \\ & ~~~~~~~~~~~~~~\forall k=0,\dots,T-1,\\
& \left[\left[ \varphi\right]\right](\xi^{\ast}_{\bm\cdot; \hat{x}_{0}, u}, 0)\ge\hat{\beta},  \label{constraint_e}   
\end{align} 
where the time index $\ell$ is initially set as 0, \(u^i[\ell:\ell+T-1] = \{u^i[\ell], \cdots, u^i[\ell+T-1]\}\) is the control input signal for agent $i$ and the input values are constrained to $[u^i_{\textrm{min}}, u^i_{\textrm{max}}]$, $\xi^{\ast}_{\bm\cdot; \hat{x}_{0}, u}$ is the nominal trajectory of the aggregated state starting from $\hat{x}_{0}$ with input $u=[(u^{1})^T,\dots,(u^{N})^T]^T$.

At each time index $\ell=jT$ ($j=0,1,\dots$), the local hubs send the state estimates to the cloud, and we modify the MTL formula as in (\ref{update_phi}) (Line \ref{update_MTL}). The MILP is solved for time $\ell=jT$ with the updated state values and the modified MTL formula $[\varphi]^{\ell}_{0}$ (Line \ref{recompute}). The previously computed control inputs are replaced by the newly computed control inputs from time index $\ell$ to $\ell+T-1$ (Line \ref{replace}). The same procedure repeats until a set maximal time $T_{\textrm{max}}$ is reached.

\section{Implementation}
In this section, we implement our differentially private controller synthesis approach on the example in Fig. \ref{fig_DF} (in Section \ref{sec:intro}). The nominal deterministic dynamics of the $i$th ($i=1,2$) Baxter-On-Wheels robot can be expressed as \cite{zhe_advisory}
\begin{align}
\begin{aligned}
\dot{\varrho}^{i} &= \frac{v_{\textrm{r}}^{i} + v_{\textrm{l}}^{i}}{2} \cos(\theta^{i}) = v^{i} \cos(\theta^{i}), \\
\dot{\kappa}^{i} &= \frac{v_{\textrm{r}}^{i} + v_{\textrm{l}}^{i}}{2} \sin(\theta^{i}) = v^{i} \sin(\theta^{i}),\\
\dot{\theta}_{i} &= \frac{v_{\textrm{r}}^{i} - v_{\textrm{l}}^{i}}{2d} = \omega^{i}, 
\label{eq:sys-dyn-unicycle}
\end{aligned}
\end{align}                                                                                           
where $\varrho^i$, $\kappa^i$ and $\theta^{i}$ denote the $x$-position, $y$-position and the orientation of the wheelchair base, $v_{\rm{r}}^{i}$ and $v_{\rm{l}}^{i}$ are the
wheel speeds of the right and left wheels of
the $i$th robot, respectively, \(v^{i}\) and \(\omega^{i}\) are the linear and angular velocities of
the $i$th robot, respectively, and \(d^{i}\) is the distance of any one wheel from the center of the robot base of
the $i$th robot. We feedback linearize the system as follows:

\begin{align}
\begin{aligned}
\begin{bmatrix}
\ddot{\varrho}^{i} \\ \ddot{\kappa}^{i}
\end{bmatrix} =
\begin{bmatrix}
\cos(\theta^{i}) & -\sin(\theta^{i}) \\ \sin(\theta^{i}) &
\cos(\theta^{i})
\end{bmatrix}
\begin{bmatrix}
\dot{v}^{i} \\ v^{i} \omega^{i}
\end{bmatrix}.
\end{aligned}
\label{equ:feedlinear}
\end{align}
We choose the intermediate control inputs to the robot to be \(\dot{v}^{i}\) and \(v^{i}\omega^{i}\) such that
\begin{align}
\begin{bmatrix}
\dot{v}^{i} \\ v^{i} \omega^{i}
\end{bmatrix} = \begin{bmatrix} \cos(\theta_{i}) & \sin(\theta_{i}) \\
-\sin(\theta_{i}) & \cos(\theta_{i})
\end{bmatrix}
\begin{bmatrix}
u^{i}_{1} \\ u^{i}_{2}
\end{bmatrix},
\label{eq:non-linear-ctrl}
\end{align} 
where \(u^{i}_{1}\) and \(u^{i}_{2}\) are the new control
inputs to be determined.

With (\ref{equ:feedlinear}), (\ref{eq:non-linear-ctrl}) and adding the process noise, we have the following stochastic control system with linear dynamics:
\begin{align}
\begin{bmatrix}
d^2\varrho^{i}/dt^2 \\ d^2\kappa^{i}/dt^2
\end{bmatrix} =
\begin{bmatrix}
u^{i}_{1} \\ u^{i}_{2}
\end{bmatrix}+\begin{bmatrix}
b^{i}_{1} \\ b^{i}_{2}
\end{bmatrix}dw/dt,
\end{align}
where $b^{i}_{1}=b^{i}_{2}=0.01$ ($i=1,2$).

We denote the aggregated state by $x=[\varrho^{1} ~\kappa^{1} ~\dot{\varrho}^{1}~ \dot{\kappa}^{1}~ \varrho^{2} ~$ $\kappa^{2} ~\dot{\varrho}^{2}~ \dot{\kappa}^{2}]^T$ and the output $y=x$.  We add privacy noise to $y$ using the Gaussian mechanism with the privacy parameters $\epsilon\in[\log(6), \log(10)]$ and $\delta\in[0.1, 0.4]$. We choose $M$ as the identity matrix and design $u=Kx+\zeta$, where 
\begin{align}\nonumber
\begin{split}                                            
&K=
\begin{bmatrix}
-1&0&0&0&  0&0&0&0\\
0&-1&0&0&  0&0&0&0\\
0&0&0&0 &  -1&0&0&0\\
0&0&0&0&   0&-1&0&0
\end{bmatrix},
\end{split}
\end{align}	
and $\zeta(\bm\cdot)$ is the new input signal.

We use the following MTL specification:                                                                      
\[
\begin{split}
\varphi=&\big(\Box_{[0,20]}\Diamond_{[0,10]}(x^1\in Green_1)\big) \\&\wedge \big(\Box_{[0,20]}\Diamond_{[0,10]}(x^1\in Green_4)\big) \\&\wedge \big(\Box_{[0,20]}\Diamond_{[0,10]}(x^2\in Green_2)\big) \\&\wedge\big(\Box_{[0,20]}\Diamond_{[0,10]}(x^2\in Green_3)\big)\\  
&\wedge \big(\Box_{[0,20]}\lnot Collide (x^1, x^2)\big),
\end{split}
\]
where $x^1=[\varrho^{1} ~\kappa^{1}]^T$ and $x^2=[\varrho^{2} ~\kappa^{2}]^T$ are the positions of the two agents, respectively. The regions $Green_1$, $Green_2$, $Green_3$ and $Green_4$ are square regions with the side length of 10 centered at $[0~0]^T$, $[100~0]^T$, $[0~100]^T$ and $[100~100]^T$, respectively.

%$\varphi$ means ``Robot 1 should reach both $Green_1$ and $Green_4$ at least once in every 10 consecutive time units, Robot 2 should reach both $Green_2$ and $Green_3$ at least once in every 10 consecutive time units, while always avoiding collision of the two robots''. 

The initial positions of the two agents are $x^1[0]=[0~0]^T$ and $x^2[0]=[100~0]^T$, respectively. The initial velocities are $v^1[0]=\omega^1[0]=0$ and $v^2[0]=\omega^2[0]=0$, respectively. We set $u^i_{\textrm{max}}=50$, $u^i_{\textrm{min}}=-50$ for $i=1,2$, and we set $T=10$ and $T_{\textrm{max}}=20$. We compute the control inputs for the two robots such that the MTL specification $\varphi$ is satisfied with probability at least $90\%$ (we choose $\gamma=\eta=95\%$) with minimal control efforts. Fig. \ref{plot1} shows the obtained optimal input signals $u^1$ and $u^2$.

\begin{figure}
	\centering
	\includegraphics[scale=0.2]{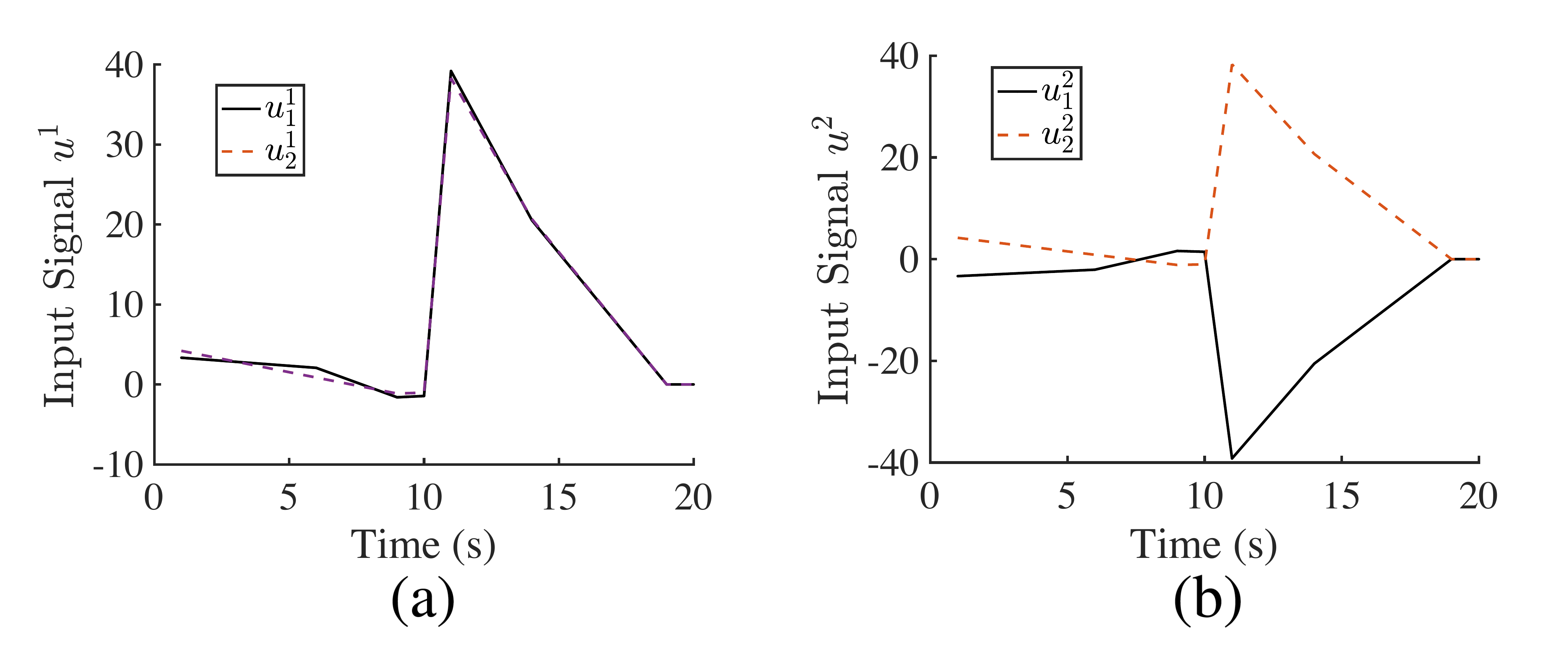}
	\caption{Obtained optimal input signals for MTL specification $\varphi$.}  
	\label{plot1}
\end{figure}

\section{Conclusion}
We presented a provably correct method for differentially private controller synthesis with respect to metric temporal logic (MTL) specifications. For future work, we will
extend the implementations to more complicated MTL specifications
and experiments on a hardware testbed.

\section*{APPENDIX}
\textbf{Proof of Theorem \ref{gamma}}:\\
To prove Theorem \ref{gamma}, we first prove that Theorem \ref{gamma} holds for any atomic proposition $\pi$.   

As the metric $d$ satisfies the triangle inequality, for any given $\tau\ge0$, we have $\forall y\in\mathcal{X}, t\ge\tau$,
\begin{align}
\begin{split}
& \vert d(\xi_{t-\tau,x_{\tau},u},y)-d(\xi^{\ast}_{t-\tau,\hat{x}_{\tau},u},y)\vert \\& \le\vert d(\xi_{t-\tau,x_{\tau},u},y)-d(\xi^{\ast}_{t-\tau,x_{\tau},u},y)\vert\\&~~+\vert d(\xi^{\ast}_{t-\tau,x_{\tau},u},y)-d(\xi^{\ast}_{t-\tau,\hat{x}_{\tau},u},y)\vert
\\&\le d(\xi_{t-\tau,x_{\tau},u},\xi^{\ast}_{t-\tau,x_{\tau},u})+d(\xi^{\ast}_{t-\tau,x_{\tau},u},\xi^{\ast}_{t-\tau,\hat{x}_{\tau},u}).
\end{split}          
\label{tri}                                                               
\end{align}                                                                    

As $\mathbb{P}[\big(\xi^{\ast}_{t-\tau,\hat{x}_{\tau},u}-\xi^{\ast}_{t-\tau,x_{\tau},u}\big)^TM\big(\xi^{\ast}_{t-\tau,\hat{x}_{\tau},u}-\xi^{\ast}_{t-\tau,x_{\tau},u}\big)<$ $\beta]>\gamma$, we have
\begin{align}
\begin{split}
&\mathbb{P}[\underbrace{d(\xi^{\ast}_{t-\tau,\hat{x}_{\tau},u},\xi^{\ast}_{t-\tau,x_{\tau},u})<\sqrt{\beta}\norm{M}^{-\frac{1}{2}}}_{\mathcal{A}}]> \gamma.
\end{split}      
\label{dist}                         
\end{align} 
On the other hand, as 
\begin{align}
\begin{split}
&\mathbb{P}[\big(\xi_{t-\tau,x_{\tau},u}-\xi^{\ast}_{t-\tau,x_{\tau},u}\big)^TM\big(\xi_{t-\tau,x_{\tau},u}-\xi^{\ast}_{t-\tau,x_{\tau},u}\big)\\&<\frac{\alpha(t-\tau)}{1-\eta}] > \eta.
\end{split}
\label{prob1}                                        
\end{align} 
Thus, we have
\begin{align}
\begin{split}
\mathbb{P}&[\underbrace{d(\xi_{t-\tau,x_{\tau},u},\xi^{\ast}_{t-\tau,x_{\tau},u})
	<\sqrt{\frac{\alpha(t-\tau)}{1-\eta}}\norm{M}^{-\frac{1}{2}}}_{\mathcal{B}}]>\eta.
\end{split}
\label{prob2}                                      
\end{align} 

From (\ref{dist}) and (\ref{prob2}), as event $\mathcal{A}$ and event $\mathcal{B}$ are independent, we have 
\begin{align}
\begin{split}
\mathbb{P}[&d(\xi^{\ast}_{t-\tau,\hat{x}_{\tau},u},\xi^{\ast}_{t-\tau,x_{\tau},u})<\sqrt{\beta}\norm{M}^{-\frac{1}{2}},\\
&d(\xi_{t-\tau,x_{\tau},u},\xi^{\ast}_{t-\tau,x_{\tau},u})
<\sqrt{\frac{\alpha(t-\tau)}{1-\eta}}\norm{M}^{-\frac{1}{2}} ] > \gamma\eta.                           
\end{split}
\label{and_prob}                                        
\end{align} 

Therefore, from (\ref{and_prob}) and (\ref{tri}), we have
\begin{align}
\begin{split}
\mathbb{P}[ \vert d(\xi_{t-\tau,x_{\tau},u},y)-d(\xi^{\ast}_{t-\tau,\hat{x}_{\tau},u},y)\vert<\hat{\beta}(t) ] > \gamma\eta,                           
\end{split}
\label{conclude_prob}                                        
\end{align} 
where $\hat{\beta}(t)\triangleq\Big(\sqrt{\beta}+\sqrt{\frac{\alpha(t-\tau)}{1-\eta}}\Big)\norm{M}^{-\frac{1}{2}}$.

In the following, we denote $B(\xi^{\ast}_{t-\tau,\hat{x}_{\tau},u},\hat{\beta}(t))\triangleq\{\xi_{t-\tau,x_{\tau},u},y)~\vert~\vert d(\xi_{t-\tau,x_{\tau},u},y)-d(\xi^{\ast}_{t-\tau,\hat{x}_{\tau},u},y)\vert<\hat{\beta}(t)\}$.

1) $\xi^{\ast}_{t-\tau,\hat{x}_{\tau},u}\in \mathcal{O}(\pi)$, and $B(\xi^{\ast}_{t-\tau,\hat{x}_{\tau},u},\hat{\beta}(t))\subset\mathcal{O}(\pi)$, as shown in Fig. \ref{proof} (a). In this case, for any $\xi_{t-\tau,x_{\tau},u}\in B(\xi^{\ast}_{t-\tau,\hat{x}_{\tau},u},\hat{\beta}(t))$, 
\begin{align}\nonumber
&\left[\left[\pi\right]\right](\xi_{\bm\cdot,x_{\tau},u}, t)=\mbox{inf}\{d(\xi_{t-\tau,x_{\tau},u},y)|y \in \mathcal{X}\backslash\mathcal{O}(\pi)\}.                                     
\end{align}
From (\ref{conclude_prob}), it holds with probability at least $\gamma\eta$ that
\begin{align}\nonumber
&\left[\left[\pi\right]\right](\xi_{\bm\cdot,x_{\tau},u}, t)\ge\mbox{inf}\{d(\xi^{\ast}_{t-\tau,\hat{x}_{\tau},u},y)-\hat{\beta}(t)|y \in \mathcal{X}\backslash\mathcal{O}(\pi)\}\\\nonumber
& =\mbox{inf}\{d(\xi^{\ast}_{t-\tau,\hat{x}_{\tau},u},y)|y \in \mathcal{X}\backslash\mathcal{O}(\pi)\}-\hat{\beta}(t)  \\\nonumber  & =\left[\left[\pi\right]\right](\xi^{\ast}_{\bm\cdot,\hat{x}_{\tau},u}, t)-\hat{\beta}(t).                               
\end{align}			

2) $\xi^{\ast}_{t-\tau,\hat{x}_{\tau},u}\notin \mathcal{O}(\pi)$, and $B(\xi^{\ast}_{t-\tau,\hat{x}_{\tau},u},\hat{\beta}(t))\subset\mathcal{X}\backslash\mathcal{O}(\pi)$, as shown in Fig. \ref{proof} (b). In this case, for any $\xi_{t-\tau,x_{\tau},u}\in B(\xi^{\ast}_{t-\tau,\hat{x}_{\tau},u},\hat{\beta}(t))$,
\begin{align}\nonumber
&\left[\left[\pi\right]\right](\xi_{\bm\cdot,x_{\tau},u}, t)=-\mbox{inf}\{d(\xi_{t-\tau,x_{\tau},u},y)|y \in cl(\mathcal{O}(\pi))\}.                                     
\end{align}
From (\ref{conclude_prob}), it holds with probability at least $\gamma\eta$ that
\begin{align}\nonumber
&\left[\left[\pi\right]\right](\xi_{\bm\cdot,x_{\tau},u}, t)\ge-\mbox{inf}\{d(\xi^{\ast}_{t-\tau,\hat{x}_{\tau},u},y)+\hat{\beta}(t)|y \in cl(\mathcal{O}(\pi))\}\\\nonumber
& =\left[\left[\pi\right]\right](\xi^{\ast}_{\bm\cdot,\hat{x}_{\tau},u}, t)-\hat{\beta}(t).                                    
\end{align}		

\begin{figure}
	\centering
	\includegraphics[scale=0.23]{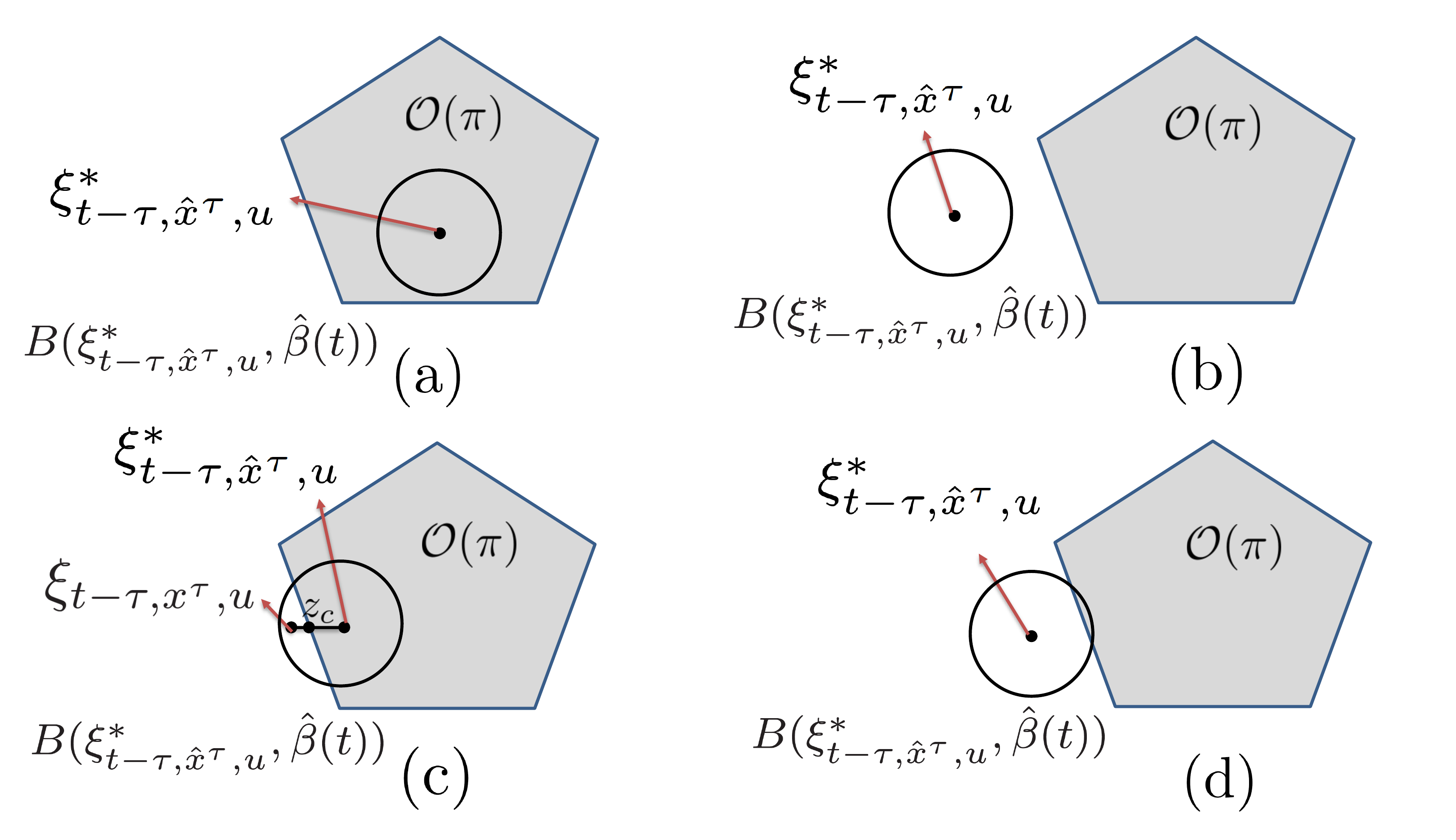}
	\caption{Four different cases in the proof.} 
	\label{proof}
\end{figure}	

3) $\xi^{\ast}_{t-\tau,\hat{x}_{\tau},u}\in \mathcal{O}(\pi)$, but $B(\xi^{\ast}_{t-\tau,\hat{x}_{\tau},u},\hat{\beta}(t))\not\subset\mathcal{O}(\pi)$, as shown in Fig. \ref{proof} (c). In this case,  it holds with probability at least $\gamma\eta$ that
\begin{align} \nonumber
\begin{split}
&\left[\left[\pi\right]\right](\xi_{\bm\cdot,x_{\tau},u}, t)\ge\min\limits_{\xi_{t-\tau,x_{\tau},u}\in B(\xi^{\ast}_{t-\tau,\hat{x}_{\tau},u},\hat{\beta}(t))}\left[\left[\pi\right]\right](\xi_{\bm\cdot,x_{\tau},u}, t)\\
&~~~~~~~~~~~~~~~~~~~~~~=\min\{X_1, X_2\}, \\
& \mbox{where}\\
&X_1=\\&-\max_{\substack{\xi_{t-\tau,x_{\tau},u}\in B(\xi^{\ast}_{t-\tau,\hat{x}_{\tau},u},\hat{\beta}(t)),\\\xi_{t-\tau,x_{\tau},u}\notin \mathcal{O}(\pi)}}\mbox{inf}\{d(\xi_{t-\tau,x_{\tau},u},y)|y \in cl(\mathcal{O}(\pi))\},\\
&X_2=\\&\min_{\substack{\xi_{t-\tau,x_{\tau},u}\in B(\xi^{\ast}_{t-\tau,\hat{x}_{\tau},u},\hat{\beta}(t)),\\\xi_{t-\tau,x_{\tau},u}\in \mathcal{O}(\pi)}}\mbox{inf}\{d(\xi_{t-\tau,x_{\tau},u},y)|y \in \mathcal{X}\backslash\mathcal{O}(\pi)\}.
\end{split}
\end{align}      

As $d(\xi_{t-\tau,x_{\tau},u},y)\ge0$, so $X_1\le0, X_2\ge0$, $\min\{X_1, X_2\}=X_1$. For any $\xi_{t-\tau,x_{\tau},u}\in B(\xi^{\ast}_{t-\tau,\hat{x}_{\tau},u},\hat{\beta}(t))$ and $\xi_{t-\tau,x_{\tau},u}\notin \mathcal{O}(\pi)$, there exists $z_c\in B(\xi^{\ast}_{t-\tau,\hat{x}_{\tau},u},\hat{\beta}(t))$ and $z_c\in \partial(\mathcal{O}(\pi))$ such that $\xi_{t-\tau,x_{\tau},u}, z_c$ and $\xi^{\ast}_{t-\tau,\hat{x}_{\tau},u}$ are collinear, i.e.
\begin{align}\nonumber
&d(\xi^{\ast}_{t-\tau,\hat{x}_{\tau},u},z_c)+d(z_c,\xi_{t-\tau,x_{\tau},u})\\ &=d(\xi^{\ast}_{t-\tau,\hat{x}_{\tau},u},\xi_{t-\tau,x_{\tau},u})\le\hat{\beta}(t).
\end{align}		
Therefore, as 
$\left[\left[\pi\right]\right](\xi^{\ast}_{\bm\cdot,\hat{x}_{\tau},u}, t)=\mbox{inf}\{d(\xi^{\ast}_{t-\tau,\hat{x}_{\tau},u},y)|y \in \mathcal{X}\backslash\mathcal{O}(\pi)\}\le d(\xi^{\ast}_{t-\tau,\hat{x}_{\tau},u},z_c)$ and $\mbox{inf}\{d(\xi_{t-\tau,x_{\tau},u},y)|y \in cl(\mathcal{O}(\pi))\}\le d(\xi_{t-\tau,x_{\tau},u},z_c)$, we have for any $x\in B(\hat{x}_{\tau},\hat{\beta}(t))$ and $\xi_{t-\tau,x_{\tau},u}\notin \mathcal{O}(\pi)$,
\begin{align}\nonumber
\mbox{inf}\{d(\xi_{t-\tau,x_{\tau},u},y)|y \in cl(\mathcal{O}(\pi))\}+\left[\left[\pi\right]\right](\xi^{\ast}_{\bm\cdot,\hat{x}_{\tau},u}, t)\le\hat{\beta}(t).                                     
\end{align}

So $-X_1+\left[\left[\pi\right]\right](\xi^{\ast}_{\bm\cdot,\hat{x}_{\tau},u}, t)\le\hat{\beta}(t)$, i.e. $X_1\ge\left[\left[\pi\right]\right](\xi^{\ast}_{\bm\cdot,\hat{x}_{\tau},u}, t)-\hat{\beta}(t)$. Therefore, it holds with probability at least $\gamma\eta$ that $\left[\left[\pi\right]\right](\xi_{\bm\cdot,x_{\tau},u}, t)\ge\min\{X_1, X_2\}=X_1\ge\left[\left[\pi\right]\right](\xi^{\ast}_{\bm\cdot,\hat{x}_{\tau},u}, t)-\hat{\beta}(t)$.

4) $\xi^{\ast}_{t-\tau,\hat{x}_{\tau},u}\notin \mathcal{O}(\pi)$, but  $B(\xi^{\ast}_{t-\tau,\hat{x}_{\tau},u},\hat{\beta}(t))\not\subset\mathcal{X}\backslash\mathcal{O}(\pi)$, as shown in Fig. \ref{proof} (d). In this case, we have
\begin{align}\nonumber
&\left[\left[\pi\right]\right](\xi^{\ast}_{\bm\cdot,\hat{x}_{\tau},u}, t)=-\mbox{inf}\{d(\xi^{\ast}_{t-\tau,\hat{x}_{\tau},u},y)|y \in cl(\mathcal{O}(\pi))\}.                                     
\end{align}		  
For any $\xi_{t-\tau,x_{\tau},u}\in B(\xi^{\ast}_{t-\tau,\hat{x}_{\tau},u},\hat{\beta}(t))$ and $\xi_{t-\tau,x_{\tau},u}\notin\mathcal{O}(\pi)$, it holds with probability at least $\gamma\eta$ that
\begin{align}\nonumber
&\left[\left[\pi\right]\right](\xi_{\bm\cdot,x_{\tau},u}, t)=-\mbox{inf}\{d(\xi_{t-\tau,x_{\tau},u},y)|y \in cl(\mathcal{O}(\pi))\}\\\nonumber
& \ge-\mbox{inf}\{d(\xi^{\ast}_{t-\tau,\hat{x}_{\tau},u},y)+\hat{\beta}(t)|y \in cl(\mathcal{O}(\pi))\}  \\\nonumber
& =\left[\left[\pi\right]\right](\xi^{\ast}_{\bm\cdot,\hat{x}_{\tau},u}, t)-\hat{\beta}(t).                                 
\end{align}			
Therefore, it holds with probability at least $\gamma\eta$ that $\left[\left[\pi\right]\right](\xi_{\bm\cdot,x_{\tau},u}, t)\ge\min\{X_1, X_2\}=X_1\ge\left[\left[\pi\right]\right](\xi^{\ast}_{\bm\cdot,\hat{x}_{\tau},u}, t)-\hat{\beta}(t)$.

In sum, we have proven that $\left[\left[\pi\right]\right](\xi_{\bm\cdot,x_{\tau},u}, t)\ge\left[\left[\pi\right]\right](\xi^{\ast}_{\bm\cdot,\hat{x}_{\tau},u}, t)-\hat{\beta}(t)$ holds with probability at least $\gamma\eta$. Similarly, we can prove that $\left[\left[\pi\right]\right](\xi_{\bm\cdot,x_{\tau},u}, t)\le \left[\left[\pi\right]\right](\xi^{\ast}_{\bm\cdot,\hat{x}_{\tau},u}, t)+\hat{\beta}(t)$ holds with probability at least $\gamma\eta$.

Therefore, Theorem \ref{gamma} holds for any atomic proposition $\pi$. Next, we use induction to prove that Theorem \ref{gamma} holds for any MTL formula $\varphi$.

(ii) We assume that Theorem \ref{gamma} holds for $\varphi$ and prove Theorem \ref{gamma} holds for $\lnot\varphi$.  

If Theorem \ref{gamma} holds for $\varphi$, then as $\left[\left[\lnot\varphi\right]\right](\xi^{\ast}_{\bm\cdot,\hat{x}_{\tau},u}, t)=-\left[\left[\varphi\right]\right](\xi^{\ast}_{\bm\cdot,\hat{x}_{\tau},u}, t)$, it holds with probability at least $\gamma\eta$ that $-\left[\left[\lnot\varphi\right]\right](\xi^{\ast}_{\bm\cdot,\hat{x}_{\tau},u}, t)-\hat{\beta}_{\rm{max}}\le -\left[\left[\lnot\varphi\right]\right](\xi_{\bm\cdot,x_{\tau},u}, t)\le -\left[\left[\lnot\varphi\right]\right](\xi^{\ast}_{\bm\cdot,\hat{x}_{\tau},u}, t)+\hat{\beta}_{\rm{max}}$, thus it holds with probability at least $\gamma\eta$ that $\left[\left[\lnot\varphi\right]\right](\xi^{\ast}_{\bm\cdot,\hat{x}_{\tau},u}, t)-\hat{\beta}_{\rm{max}}\le \left[\left[\lnot\varphi\right]\right](\xi_{\bm\cdot,x_{\tau},u}, t)\le \left[\left[\lnot\varphi\right]\right](\xi^{\ast}_{\bm\cdot,\hat{x}_{\tau},u}, t)+\hat{\beta}_{\rm{max}}$.

(iii) We assume that Theorem \ref{gamma} holds for $\varphi_1,\varphi_2$ and prove Theorem \ref{gamma} holds for $\varphi_1\wedge\varphi_2$. 

If Theorem \ref{gamma} holds for $\varphi_1$ and $\varphi_2$, then it holds with probability at least $\gamma\eta$ that $\left[\left[\varphi_1\right]\right](\xi^{\ast}_{\bm\cdot,\hat{x}_{\tau},u}, t)-\hat{\beta}_{\rm{max}}\le \left[\left[\varphi_1\right]\right](\xi_{\bm\cdot,x_{\tau},u}, t)\le \left[\left[\varphi_1\right]\right](\xi^{\ast}_{\bm\cdot,\hat{x}_{\tau},u}, t)+\hat{\beta}_{\rm{max}}$, $\left[\left[\varphi_2\right]\right](\xi^{\ast}_{\bm\cdot,\hat{x}_{\tau},u}, t)-\hat{\beta}_{\rm{max}}\le \left[\left[\varphi_2\right]\right](\xi_{\bm\cdot,x_{\tau},u}, t)\le \left[\left[\varphi_2\right]\right](\xi^{\ast}_{\bm\cdot,\hat{x}_{\tau},u}, t)+\hat{\beta}_{\rm{max}}$. As $\left[\left[\varphi_1\wedge\varphi_2\right]\right](\xi^{\ast}_{\bm\cdot,\hat{x}_{\tau},u}, t)=\min(\left[\left[\varphi_1\right]\right](\xi^{\ast}_{\bm\cdot,\hat{x}_{\tau},u}, t),\left[\left[\varphi_2\right]\right](\xi^{\ast}_{\bm\cdot,\hat{x}_{\tau},u}, t))$, it holds with probability at least $\gamma\eta$ that
\begin{align}\nonumber
\begin{split}
&\min(\left[\left[\varphi_1\right]\right](\xi^{\ast}_{\bm\cdot,\hat{x}_{\tau},u}, t),\left[\left[\varphi_2\right]\right](\xi^{\ast}_{\bm\cdot,\hat{x}_{\tau},u}, t))-\hat{\beta}_{\rm{max}}\\&\le\left[\left[\varphi_1\wedge\varphi_2\right]\right](\xi_{\bm\cdot,x_{\tau},u}, t)\le \min(\left[\left[\varphi_1\right]\right](\xi^{\ast}_{\bm\cdot,\hat{x}_{\tau},u},t),\\
& \left[\left[\varphi_2\right]\right](\xi^{\ast}_{\bm\cdot,\hat{x}_{\tau},u}, t))+\hat{\beta}_{\rm{max}},
\end{split}
\end{align}
therefore it holds with probability at least $\gamma\eta$ that $\left[\left[\varphi_1\wedge\varphi_2\right]\right](\xi^{\ast}_{\bm\cdot,\hat{x}_{\tau},u}, t)-\hat{\beta}_{\rm{max}}\le$ $\left[\left[\varphi_1\wedge\varphi_2\right]\right](\xi_{\bm\cdot,x_{\tau},u}, t)\le \left[\left[\varphi_1\wedge\varphi_2\right]\right](\xi^{\ast}_{\bm\cdot,\hat{x}_{\tau},u}, t)+\hat{\beta}_{\rm{max}}$.                                            

(iv) We assume that Theorem \ref{gamma} holds for $\varphi$ and prove Theorem \ref{gamma} holds for $\varphi_1\mathcal{U}_{\mathcal{I}}\varphi_2$.   

As
\begin{align}\nonumber 
\begin{split}
\left[\left[\varphi_1\mathcal{U}_{\mathcal{I}}\varphi_{2}\right]\right](\xi^{\ast}_{\bm\cdot,\hat{x}_{\tau},u}, t)=&\max\limits_{t'\in (t+\mathcal{I})}\Big(\min\big(\left[\left[\varphi_2\right]\right](\xi^{\ast}_{\bm\cdot,\hat{x}_{\tau},u}, t'), \\& \min\limits_{t\le t''<t'}\left[\left[\varphi_1\right]\right]
(\xi^{\ast}_{\bm\cdot,\hat{x}_{\tau},u},t'')\big)\Big), 
\end{split}
\end{align}
if Theorem \ref{gamma} holds for $\varphi_1$ and $\varphi_2$, then it holds with probability at least $\gamma\eta$ that $\left[\left[\varphi_1\right]\right](\xi^{\ast}_{\bm\cdot,\hat{x}_{\tau},u}, t'')-\hat{\beta}_{\rm{max}}\le \left[\left[\varphi_1\right]\right](\xi_{\bm\cdot,x_{\tau},u}, t'')\le \left[\left[\varphi_1\right]\right](\xi^{\ast}_{\bm\cdot,\hat{x}_{\tau},u}, t'')+\hat{\beta}_{\rm{max}}$, $\left[\left[\varphi_2\right]\right](\xi^{\ast}_{\bm\cdot,\hat{x}_{\tau},u}, t')-\hat{\beta}_{\rm{max}}\le \left[\left[\varphi_2\right]\right](\xi_{\bm\cdot,x_{\tau},u}, t')\le \left[\left[\varphi_2\right]\right](\xi^{\ast}_{\bm\cdot,\hat{x}_{\tau},u}, t')+\hat{\beta}_{\rm{max}}$, so it holds with probability at least $\gamma\eta$ that
\begin{align}\nonumber            
\begin{split}
&\max_{t'\in (t+\mathcal{I})}\Big(\min\big(\left[\left[\varphi_2\right]\right](\xi^{\ast}_{\bm\cdot,\hat{x}_{\tau},u}, t'), \\
&~~~~~~~~~~~~~~\min_{t\le t''<t'}\left[\left[\varphi_1\right]\right]
(\xi^{\ast}_{\bm\cdot,\hat{x}_{\tau},u},t'')\big)\Big)-\hat{\beta}_{\rm{max}}\\
&\le\max_{t'\in (t+\mathcal{I})}\Big(\min\big(\left[\left[\varphi_2\right]\right](\xi_{\bm\cdot,x_{\tau},u}, t'), \min_{t\le t''<t'}\left[\left[\varphi_1\right]\right]
(\xi_{\bm\cdot,x_{\tau},u},t'')\big)\Big)\\
&\le\max_{t'\in (t+\mathcal{I})}\Big(\min\big(\left[\left[\varphi_2\right]\right](\xi^{\ast}_{\bm\cdot,\hat{x}_{\tau},u}, t'), \\
&~~~~~~~~~~~~~~\min_{t\le t''<t'}\left[\left[\varphi_1\right]\right]
(\xi^{\ast}_{\bm\cdot,\hat{x}_{\tau},u},t'')\big)\Big)+\hat{\beta}_{\rm{max}}.   
\end{split}
\end{align}
Thus Theorem \ref{gamma} holds for $\varphi_1\mathcal{U}_{\mathcal{I}}\varphi_{2}$.

Therefore, it is proved by induction that Theorem \ref{gamma} holds for any MTL formula $\varphi$.

\bibliographystyle{IEEEtran}
\bibliography{DFref}

% Generated by IEEEtran.bst, version: 1.12 (2007/01/11)
\begin{thebibliography}{10}
\providecommand{\url}[1]{#1}
\csname url@samestyle\endcsname
\providecommand{\newblock}{\relax}
\providecommand{\bibinfo}[2]{#2}
\providecommand{\BIBentrySTDinterwordspacing}{\spaceskip=0pt\relax}
\providecommand{\BIBentryALTinterwordstretchfactor}{4}
\providecommand{\BIBentryALTinterwordspacing}{\spaceskip=\fontdimen2\font plus
\BIBentryALTinterwordstretchfactor\fontdimen3\font minus
  \fontdimen4\font\relax}
\providecommand{\BIBforeignlanguage}[2]{{%
\expandafter\ifx\csname l@#1\endcsname\relax
\typeout{** WARNING: IEEEtran.bst: No hyphenation pattern has been}%
\typeout{** loaded for the language `#1'. Using the pattern for}%
\typeout{** the default language instead.}%
\else
\language=\csname l@#1\endcsname
\fi
#2}}
\providecommand{\BIBdecl}{\relax}
\BIBdecl

\bibitem{zhe2019privacy}
Z.~Xu and A.~A. Julius, ``Robust temporal logic inference for provably correct
  fault detection and privacy preservation of switched systems,'' \emph{IEEE
  Systems Journal}, vol.~13, no.~3, pp. 3010--3021, 2019.

\bibitem{Dwork2014algorithmic}
C.~Dwork, A.~Roth \emph{et~al.}, ``The algorithmic foundations of differential
  privacy,'' \emph{Foundations and Trends{\textregistered} in Theoretical
  Computer Science}, vol.~9, no. 3--4, pp. 211--407, 2014.

\bibitem{LeNy2014}
J.~L. Ny and G.~J. Pappas, ``Differentially private filtering,'' \emph{IEEE
  Transactions on Automatic Control}, vol.~59, no.~2, pp. 341--354, Feb 2014.

\bibitem{Zhou2015OptimalMP}
Y.~Zhou, D.~Maity, and J.~S. Baras, ``Optimal mission planner with timed
  temporal logic constraints,'' \emph{2015 European Control Conference (ECC)},
  pp. 759--764, 2015.

\bibitem{zhe_ijcai2019}
\BIBentryALTinterwordspacing
Z.~Xu and U.~Topcu, ``Transfer of temporal logic formulas in reinforcement
  learning,'' in \emph{IJCAI-19}, 7 2019, pp. 4010--4018. [Online]. Available:
  \url{https://doi.org/10.24963/ijcai.2019/557}
\BIBentrySTDinterwordspacing

\bibitem{zheACC2018}
Z.~{Xu}, A.~{Julius}, and J.~H. {Chow}, ``Coordinated control of wind turbine
  generator and energy storage system for frequency regulation under temporal
  logic specifications,'' in \emph{2018 Annual American Control Conference
  (ACC)}, June 2018, pp. 1580--1585.

\bibitem{zhe_control}
Z.~Xu, A.~Julius, and J.~H. Chow, ``Energy storage controller synthesis for
  power systems with temporal logic specifications,'' \emph{IEEE Systems
  Journal}, vol.~13, no.~1, pp. 748--759, 2019.

\bibitem{Julius2008CDC}
A.~A. Julius and G.~J. Pappas, ``Probabilistic testing for stochastic hybrid
  systems,'' in \emph{2008 47th IEEE Conference on Decision and Control}, Dec
  2008, pp. 4030--4035.

\bibitem{FainekosMTL}
G.~E. Fainekos and G.~J. Pappas\vspace{0mm}, ``Robustness of temporal logic
  specifications,'' in \emph{Formal Approaches to Testing and Runtime
  Verification, in: LNCS, vol. 4262, Springer, 2006}.

\bibitem{FAINEKOScontinous}
G.~E. Fainekos and G.~J. Pappas, ``Robustness of temporal logic specifications
  for continuous-time signals,'' \emph{Theoretical Computer Science}, vol. 410,
  no.~42, pp. 4262 -- 4291, 2009.

\bibitem{YazdaniLQ2018}
K.~Yazdani, A.~Jones, K.~Leahy, and M.~T. Hale, ``{Differentially Private LQ
  Control},'' \emph{arXiv preprint arXiv:1807.05082}, 2018, available at:
  https://arxiv.org/abs/1807.05082.

\bibitem{hajek2015random}
B.~Hajek, \emph{Random Processes for Engineers}.\hskip 1em plus 0.5em minus
  0.4em\relax Cambridge University Press, 2015.

\bibitem{Bernstein2009}
D.~S. Bernstein, \emph{Matrix Mathematics: Theory, Facts, and Formulas: Second
  Edition}, 2nd~ed.\hskip 1em plus 0.5em minus 0.4em\relax Princeton University
  Press, 2009.

\bibitem{zhe_advisory}
Z.~Xu, S.~Saha, B.~Hu, S.~Mishra, and A.~Julius, ``Advisory temporal logic
  inference and controller design for semiautonomous robots,'' \emph{IEEE
  Trans. Autom. Sci. Eng.}, pp. 1--19, 2018.

\end{thebibliography}

\end{document}